\documentclass{sig-alternate}
\toappear{}

\usepackage[ruled, vlined, linesnumbered, nofillcomment]{algorithm2e}

\usepackage[english]{babel}
\usepackage{amssymb}
\usepackage{amsmath}
\usepackage[amsmath]{ntheorem}
\usepackage{subfigure}
\usepackage{booktabs}
\usepackage{cite}
\usepackage{url}
\usepackage{graphicx}
\usepackage{mdwlist}
\usepackage{tikz}
\usetikzlibrary{snakes,arrows,shapes,positioning,fit,calc}

\usepackage{ifthen}
\newcommand{\set}[1]{\left\{#1\right\}}
\newcommand{\pr}[1]{\left(#1\right)}
\newcommand{\fpr}[1]{\mathopen{}\left(#1\right)}

\newcommand{\abs}[1]{{\left|#1\right|}}

\newcommand{\enset}[2]{\left\{#1 ,\ldots , #2\right\}}
\newcommand{\enpr}[2]{\pr{#1 ,\ldots , #2}}

\newcommand{\np}{\textbf{NP}}
\newcommand{\p}{\textbf{P}}

\newcommand{\define}{\leftarrow}

\newcommand{\dispfunc}[2]{{\mathit{#1}\fpr{#2}}}

\newcommand{\sources}[1]{\dispfunc{src}{#1}}

\newcommand{\we}[1]{\dispfunc{WE}{#1}}
\newcommand{\pe}[1]{\dispfunc{PE}{#1}}
\newcommand{\range}[1]{\dispfunc{range}{#1}}
\newcommand{\inst}[1]{\dispfunc{inst}{#1}}
\newcommand{\pre}[1]{\dispfunc{pre}{#1}}
\newcommand{\tail}[1]{\dispfunc{tail}{#1}}

\newcommand{\icl}[1]{\dispfunc{icl}{#1}}
\newcommand{\tcl}[1]{\dispfunc{tcl}{#1}}
\newcommand{\lab}[1]{\dispfunc{lab}{#1}}
\newcommand{\ts}[1]{\dispfunc{ts}{#1}}
\newcommand{\node}[1]{\dispfunc{nd}{#1}}
\newcommand{\id}[1]{\dispfunc{id}{#1}}

\newcommand{\lexleq}{\leq_{\mathit{lex}}}
\newcommand{\first}[1]{\dispfunc{first}{#1}}
\newcommand{\last}[1]{\dispfunc{last}{#1}}

\newcommand{\efam}[1]{\mathcal{#1}}

\newcommand{\freq}[1]{\dispfunc{fr}{#1}}

\SetKwComment{tcpas}{\{}{\}}
\SetCommentSty{textnormal}
\SetArgSty{textnormal}
\SetKw{Break}{break}
\SetKw{Continue}{continue}
\SetKw{False}{false}
\SetKw{True}{true}
\SetKw{Undecided}{undecided}
\SetKw{Null}{null}
\SetKw{Output}{output}
\SetKw{AND}{and}
\SetKw{OR}{or}
\SetKwInOut{Input}{input}
\SetKwInOut{Out}{output}

\theoremheaderfont{\normalfont\scshape}
\newtheorem{theorem}{Theorem}
\newtheorem{lemma}[theorem]{Lemma}

\theorembodyfont{\normalfont}
\newtheorem{definition}[theorem]{Definition}
\newtheorem{example}[theorem]{Example}

\pgfdeclarelayer{background}
\pgfdeclarelayer{foreground}
\pgfsetlayers{background,main,foreground}

\begin{document}

\title{Mining Closed Episodes with Simultaneous Events}
\numberofauthors{2}
\author{
Nikolaj Tatti, Boris Cule\\
	\affaddr{ADReM, University of Antwerp, Antwerpen, Belgium} \\
	\email{firstname.lastname@ua.ac.be}
}

\maketitle
\begin{abstract}
Sequential pattern discovery is a well-studied field in data mining.
Episodes are sequential patterns describing events that often occur in the
vicinity of each other. Episodes can impose restrictions to the order of the
events, which makes them a versatile technique for describing complex patterns
in the sequence. Most of the research on episodes deals with special
cases such as serial, parallel, and injective episodes, while discovering general
episodes is understudied.

In this paper we extend the definition of an episode in order to be able to
represent cases where events often occur simultaneously. We present an
efficient and novel miner for discovering frequent and closed general episodes.
Such a task presents unique challenges. Firstly, we cannot define closure based on frequency.
We solve this by computing a more conservative closure that we use to reduce the search space
and discover the closed episodes as a postprocessing step.
Secondly, episodes are traditionally presented as directed acyclic graphs. We argue that
this representation has drawbacks leading to redundancy in the output. We solve
these drawbacks by defining a subset
relationship in such a way that allows us to remove the redundant episodes.
We demonstrate the efficiency of our algorithm and the need for using closed
episodes empirically on synthetic and real-world datasets.

\end{abstract}




\section{Introduction}
\label{sec:introduction}

Discovering interesting patterns in data sequences is a popular aspect of
data mining. An episode is a sequential pattern representing a set of events
that reoccur in a sequence~\cite{mannila:97:discovery}. In its most general
form, an episode also imposes a partial order on the events. This allows great
flexibility in describing complex interactions between the events in the sequence.

Existing research in episode mining is dominated by two special cases: parallel episodes,
patterns where the order of the events does not matter, and serial episodes, 
requiring that the events must occur in one given order. Proposals have been
made (see Section~\ref{sec:related}) for discovering episodes with partial
orders, but these approaches impose various limitations on the events.  In fact,
to our knowledge, there is no published work giving an explicit description of a
miner for general episodes.

We believe that there are two main reasons why general episodes have attracted
less interest: Firstly, implementing a miner is surprisingly difficult:  
testing whether an episode occurs in the sequence is an \np-complete
problem. Secondly, the fact that episodes are such
a rich pattern type leads to a severe pattern explosion.

Another limitation of episodes is that they do not properly address
simultaneous events. However, sequences containing such events are frequently encountered, in cases such as, for example,
sequential data generated by multiple sensors and then collected into one
stream. In such a setting, if two events, say $a$ and $b$, often occur simultaneously, existing approaches will depict this pattern as a parallel episode $\{a,b\}$,
which will only tell the user that these two events often occur near each other, in no particular order.
This is a major limitation, since the actual pattern contains much more information.

In this paper we propose a novel and practical algorithm for mining
frequent closed episodes that properly handles simultaneous events. Such a task
poses several challenges.

Firstly, we can impose four different relationships between two events $a$ and
$b$: (1) the order of $a$ and $b$ does not matter, (2) events $a$ and $b$ should
occur at the same time, (3) $b$ should occur after $a$, and (4) $b$ should
occur after or at the same time as $a$. We extend the definition of an episode
to handle all these cases. In further text, we consider events simultaneous only if they
occur exactly at the same time. However, we can easily adjust
our framework to consider events simultaneous if they occur within
a chosen time interval.

Secondly, a standard approach for representing a partial order of the events
is by using a directed acyclic graph (DAG). The mining algorithm would then
discover episodes by adding nodes and edges. However, we point out that such a
representation has drawbacks. One episode may be represented by several graphs
and the subset relationship based on the graphs is not optimal. This ultimately
leads to outputting redundant patterns. We will address this problem.

Thirdly, we attack the problem of pattern explosion by using closed patterns.
There are two particular challenges with closed episodes. Firstly, we point out
that we cannot define a unique closure for an
episode, that is, an episode may have several maximal episodes with the same
frequency.  Secondly, the definition of a closure requires a subset
relationship, and  computing the subset relationship between episodes is \np-hard.

We mine patterns using a depth-first search. An episode is represented by a
DAG and we explore the patterns  by adding nodes and edges. To reduce the
search space we use the \emph{instance-closure} of episodes. While it is not
guaranteed that an instance-closed episode is actually closed, using such
episodes will greatly trim the pattern space. Finally, the actual filtering for
closed episodes is done in a post-processing step. We introduce techniques for
computing the subset relationship, distinguishing the cases where we can do a
simple test from the cases where we have to resort to recursive enumeration.
This filtering will remove all redundancies resulting from using DAGs for
representing episodes.

The rest of the paper is organised as follows: In Section~\ref{sec:related}, we discuss the most relevant related work. In Section~\ref{sec:episodes},
we present the main notations and concepts. Section~\ref{sec:subset} introduces
the notion of closure in the context of episodes.  Our algorithm is presented
in detail in Sections~\ref{sec:instance},~\ref{sec:algorithm}
and~\ref{sec:compute}. In Section~\ref{sec:exp} we present the results of our
experiments, before presenting our conclusions in
Section~\ref{sec:conclusions}. The proofs of the theorems can be found in the Appendix
and the code of the algorithm is available online\footnote{\url{http://adrem.ua.ac.be/implementations}}.

\section{Related Work}
\label{sec:related}
The first attempt at discovering frequent subsequences, or serial
episodes, was made by Wang et al.~\cite{wang:94:combinatorial}. The dataset
consisted of a number of sequences, and a pattern was considered interesting if
it was long enough and could be found in a sufficient number of sequences. A
complete solution to a more general problem was later provided by Agrawal and
Srikant~\cite{agrawal:95:mining} using an \textsc{Apriori}-style
algorithm~\cite{agrawal:94:fast}.

Looking for frequent general episodes in a single event sequence was first
proposed by Mannila et al.~\cite{mannila:97:discovery}. The \textsc{Winepi}
algorithm finds all episodes that occur in a sufficient number of windows of
fixed length.  Specific algorithms were given for the case of parallel and
serial episodes.  However, no algorithm for detecting general episodes 
was provided.

Some research has gone into outputting only closed subsequences, where a
sequence is considered closed if it is not properly contained in any other
sequence which has the same frequency. Yan et al.~\cite{yan:03:clospan},
Tzvetkov et al.~\cite{tzvetkov:03:mining}, and Wang and Han~\cite{wang:04:bide}
proposed methods for mining such closed patterns, while
Garriga~\cite{garriga:05:summarizing} further reduced the output by
post-processing it and representing the patterns using partial orders.\footnote{Despite their name, the partial
orders discovered by Garriga are different from general episodes.} Harms et al.~\cite{harms:01:discovering}, meanwhile, experiment with closed serial episodes.
In another attempt to trim the output, Garofalakis et al.~\cite{garofalakis:02:mining}
proposed a family of algorithms called \textsc{Spirit} which allow the user to
define regular expressions that specify the language that the discovered patterns must
belong to.

Pei et al.~\cite{pei:06:discovering}, and Tatti and Cule~\cite{tatti:10:mining}
considered restricted versions of our problem setup. The former approach
assumes a dataset of sequences where the same label can occur only once.
Hence, an episode can contain only unique labels. The latter pointed
out the problem of defining a proper subset relationship between general episodes and tackled it by
considering only episodes where two nodes having the same label had to be
connected. In our work, we impose no restrictions on the labels of events making up the episodes.

In this paper we use frequency based on a sliding window as it
is defined for \textsc{Winepi}. However, we can easily adopt our approach for
other monotonically decreasing measures, as well as to a setup where the data
consists of many (short) sequences instead of a single long one. Mannila et al. propose
\textsc{Minepi}~\cite{mannila:97:discovery}, an alternative interestingness measure for an episode, where
the support is defined as the number of minimal windows. Unfortunately, this
measure is not monotonically decreasing. However,
the issue can be fixed by defining support as the maximal number of non-overlapping minimal
windows~\cite{tatti:09:significance,laxman:07:fast}.  Zhou et
al.~\cite{zhou:10:mining} proposed mining closed serial episodes based on the
\textsc{Minepi} method. However, the paper did not address the non-monotonicity
issue of \textsc{Minepi}.

Alternative interestingness measures, either statistically motivated or aimed
to remove bias towards smaller episodes, were made by
Garriga~\cite{casas-garriga:03:discovering}, M\'eger and Rigotti~\cite{meger:04:constraint-based}, Gwadera et
al.~\cite{gwadera:05:reliable,gwadera:05:markov}, Calders et
al.~\cite{calders:07:mining}, Cule et al.~\cite{cule:09:new}, and
Tatti~\cite{tatti:09:significance}.

Using episodes to discover simultaneous events has, to our knowledge, not been
done yet. However, this work is somewhat related to efforts made in discovering
sequential patterns in multiple streams~\cite{oates:96:searching,
chen:05:sequential, gwadera:08:discovering}.  Here, it is possible to discover
a pattern wherein two events occur simultaneously, as long as they occur in
separate streams.

\section{Episodes with Simultaneous\\ Events}
\label{sec:episodes}

\tikzstyle{node} = [inner sep = 1pt]
\tikzstyle{labnode} = [inner sep = 1pt]
\tikzstyle{wedge} = [draw, -latex', thick,green!60!black!100, densely dashed]
\tikzstyle{pedge} = [draw, -latex', thick,blue!70]

We begin this section by introducing the basic concepts that we will use throughout
the paper. First we will describe our dataset.

\begin{definition}
A \emph{sequence event} $e = (\id{e}, \lab{e}, \ts{e})$ is a tuple consisting of three
entries, a unique id number $\id{e}$, a label $\lab{e}$, and a time stamp integer
$\ts{e}$. We will assume that if $\id{e} > \id{f}$, then $\ts{e} \geq \ts{f}$.
A \emph{sequence} is a collection of sequence events ordered by
their ids.
\end{definition}

Note that we are allowing multiple events to have the same time stamp even when
their labels are equivalent. For the sake of simplicity, we will use the
notation $s_1\cdots s_N$ to mean a sequence $\enpr{(1, s_1, 1)}{(N, s_N, N)}$.
We will also write $s_1\cdots(s_is_{i + 1})\cdots s_N$ to mean the sequence
\[
	((1, s_1, 1), \ldots, (i, s_i, i), (i + 1, s_{i + 1}, i), \ldots, (N, s_N, N - 1)).
\]
This means that $s_i$ and $s_{i + 1}$ have equal time stamps.

Our next step is to define patterns we are interested in.

\begin{definition}
An \emph{episode event} $e$ is a tuple consisting of two entries, a unique id
number $\id{e}$ and a label $\lab{e}$. An \emph{episode graph} $G$ is a directed
acyclic graph (DAG). The graph may have two types of edges:
\emph{weak edges} $\we{G}$ and \emph{proper edges} $\pe{G}$.

An \emph{episode} consists of a collection of episode events, an episode graph,
and a surjective mapping from episode events to the nodes of the graph which we
will denote by $\node{e}$.  A proper edge from node $v$ to node $w$ in the episode graph implies that the events of $v$ must occur before the events of $w$, while a weak edge from $v$ to $w$ implies that the events of $w$ may occur either at the same time as those of $v$ or later.

We will assume that the nodes of $G$ are indexed and we will use the notation
$\node{G, i}$ to refer to the $i$th node in $G$.
\end{definition}

When there is no danger of confusion, we will use the same letter to denote an
episode and its graph. Note that we are allowing multiple episode events to
share the same node even if these events have the same labels.

\begin{definition}
Given an episode $G$ and a node $n$, we define 
$\lab{n} = \set{\lab{e} \mid \node{e} = n}$ to be the multiset of labels associated with the node.
Given two multisets of labels $X$ and $Y$ we write $X \lexleq Y$ if $X$ is
lexicographically smaller than or equal to $Y$.
We also define $\lab{G}$ to be the multiset of all labels in $G$.
\end{definition}

\begin{definition}
A node $n$ in an episode graph is a \emph{descendant} of a node $m$ if there
is a path from $m$ to $n$. If there is a path containing a proper edge we will call $n$
a \emph{proper descendant} of $m$. We similarly define a \emph{(proper) ancestor}.
A node $n$ is a \emph{source} if it has no ancestors. A node $n$ is a \emph{proper source}
if it has no proper ancestors. We denote all sources of an episode $G$ by $\sources{G}$.
\end{definition}

We are now ready to give a precise definition of an occurrence of a pattern in a sequence.

\begin{definition}
Given a sequence $s$ and an episode $G$, we say that $s$ covers $G$ if there
exists an \emph{injective} mapping $m$ from the episode events to the sequence events such that  
\begin{enumerate*}
\item labels are respected, $\lab{m(e)} = \lab{e}$,
\item events sharing a same node map to events with the same time stamp, in other words,
$\node{e} = \node{f}$ implies $\ts{m(e)} = \ts{m(f)}$,
\item weak edges are respected, if $\node{e}$ is a descendant of $\node{f}$, then
$\ts{m(e)} \geq \ts{m(f)}$,
\item proper edges are respected, if $\node{e}$ is a proper descendant of $\node{f}$, then
$\ts{m(e)} > \ts{m(f)}$.
\end{enumerate*}
Note that this definition allows us to abuse notation and map graph nodes
directly to time stamps, that is, given a graph node $n$ we define $\ts{m(n)} =
\ts{m(e)}$, where $\node{e} = n$.
\end{definition}

Consider the first three episodes in Figure~\ref{fig:toy}.  A sequence $ab(cd)$
covers $G_1$ and $G_3$ but not $G_2$ (proper edge $(c, d)$ is violated).  A
sequence $(ab)cd$ covers $G_1$ and $G_2$ but not $G_3$. 

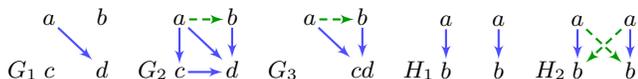
\begin{figure}[htb!]
\centering

\begin{tikzpicture}
\node[node, anchor = base] (n1) {$a$};
\node[node, right = 0.7cm of n1.base, anchor = base] (n2) {$b$};
\node[node, below = 0.7cm of n1.base, anchor = base] (n3) {$c$};
\node[node, right = 0.7cm of n3.base, anchor = base] (n4) {$d$};
\node[labnode, left = 0cm of n3.base west, anchor = base east] (lab) {$G_1$};
\path [pedge] (n1) -- (n4);
\end{tikzpicture}\hfill
\begin{tikzpicture}
\node[node, anchor = base] (n1) {$a$};
\node[node, right = 0.7cm of n1.base, anchor = base] (n2) {$b$};
\node[node, below = 0.7cm of n1.base, anchor = base] (n3) {$c$};
\node[node, right = 0.7cm of n3.base, anchor = base] (n4) {$d$};
\node[labnode, left = 0cm of n3.base west, anchor = base east] (lab) {$G_2$};
\path [pedge] (n1) -- (n4);
\path [pedge] (n1) -- (n3);
\path [wedge] (n1.mid east) -- (n2.mid west);
\path [pedge] (n3.mid east) -- (n4.mid west);
\path [pedge] (n2) -- (n4);
\end{tikzpicture}\hfill
\begin{tikzpicture}
\node[node, anchor = base] (n1) {$a$};
\node[node, right = 0.7cm of n1.base, anchor = base] (n2) {$b$};
\node[node, right = 0.7cm of n3.base, anchor = base] (n4) {$cd$};
\node[labnode, left = 0cm of n3.base west, anchor = base east] (lab) {$G_3$};
\path [pedge] (n1) -- (n4);
\path [wedge] (n1.mid east) -- (n2.mid west);
\path [pedge] (n2) -- (n4);
\end{tikzpicture}\hfill
\begin{tikzpicture}
\node[node, anchor = base] (n1) {$a$};
\node[node, right = 0.7cm of n1.base, anchor = base] (n2) {$a$};
\node[node, below = 0.7cm of n1.base, anchor = base] (n3) {$b$};
\node[node, right = 0.7cm of n3.base, anchor = base] (n4) {$b$};
\node[labnode, left = 0cm of n3] (lab) {$H_1$};
\path [pedge] (n1) -- (n3);
\path [pedge] (n2) -- (n4);
\end{tikzpicture}\hfill
\begin{tikzpicture}
\node[node, anchor = base] (n1) {$a$};
\node[node, right = 0.7cm of n1.base, anchor = base] (n2) {$a$};
\node[node, below = 0.7cm of n1.base, anchor = base] (n3) {$b$};
\node[node, right = 0.7cm of n3.base, anchor = base] (n4) {$b$};
\node[labnode, left = 0cm of n3] (lab) {$H_2$};
\path [pedge] (n1) -- (n3);
\path [pedge] (n2) -- (n4);
\path [wedge] (n2) -- (n3);
\path [wedge] (n1) -- (n4);
\end{tikzpicture}\hfill
\caption{Toy episodes. Proper edges are drawn solid. Weak edges are drawn dashed.}
\label{fig:toy}
\end{figure}

Finally, we are ready to define support of an episode based on fixed windows.
This definition corresponds to the definition used in
\textsc{Winepi}~\cite{mannila:97:discovery}. The support is monotonically decreasing which
allows us to do effective pruning while discovering frequent episodes.

\begin{definition}
Given a sequence $s$ and two integers $i$ and $j$ we define a \emph{subsequence}
\[
	s[i, j] = \set{e \in s \mid i \leq \ts{e} \leq  j}
\]
containing all events occurring between $i$ and $j$.
\end{definition}

\begin{definition}
Given a window size
$\rho$ and an episode $s$, we define the \emph{support} of an episode $G$ in $s$, denoted $\freq{G; s}$, to be the number of windows
of size $\rho$ in $s$ covering the episode,
\[
	\freq{G; s} = \abs{\set{s[i, i + \rho - 1] \mid s[i, i + \rho - 1] \text{ covers } G}}.
\]
We will use $\freq{G}$ whenever $s$ is clear from the context.
An episode is $\sigma$-\emph{frequent} (or simply \emph{frequent}) if its support is higher or equal than some given threshold $\sigma$.
\end{definition}

Consider a sequence $abcdacbd$ and set the window size $\rho = 4$. There are
$2$ windows covering episode $G_1$ (given in Figure~\ref{fig:toy}), namely
$s[1, 4]$ and $s[5, 8]$. Hence $\freq{G_1} = 2$.

\begin{theorem}
\label{thr:npcomplete}
Testing whether a sequence $s$ covers an episode $G$ is an \np-complete problem,
even if $s$ does not contain simultaneous events.
\end{theorem}

\section{Subepisode Relationship}
\label{sec:subset}
In practice, episodes are represented by DAGs and are mined by adding nodes and
edges. However, such a representation has drawbacks~\cite{tatti:10:mining}. To see
this, consider episodes $H_1$ and $H_2$ in Figure~\ref{fig:toy}. Even though
these episodes have different graphs, they are essentially the same --- both
episodes are covered by exactly the same sequences, namely all sequences
containing $abab$, $a(ab)b$, $(aa)(bb)$, $aa(bb)$, $(aa)bb$, or $aabb$. In other words, essentially the same episode may
be represented by several graphs. Moreover, using the graph subset relationship
to determine subset relationships between episodes will ultimately lead to
less efficient algorithms and redundancy in the final output. To counter these
problems we introduce a subset relationship based on coverage.

\begin{definition}
Given two episodes $G$ and $H$, we say that $G$ is a \emph{subepisode} of $H$, denoted
$G \preceq H$, if any sequence covering $H$ also covers $G$. If
$G \preceq H$ and $H \preceq G$, we say that $G$ and $H$ are \emph{similar}
in which case we will write $G \sim H$.
\end{definition}

This definition gives us the optimal definition for a subset
relationship in the following sense: if $G \npreceq H$, then there exists a sequence $s$ such
that $\freq{G; s} < \freq{H; s}$.

Consider the episodes given in Figure~\ref{fig:toy}. It follows from the definition that $H_1 \sim H_2$, $G_1
\preceq G_2$, and $G_1 \preceq G_3$. Episodes $G_2$ and $G_3$ are not comparable.

\begin{theorem}
\label{thr:hardness}
Testing $G \preceq H$ is an \np-hard problem.
\end{theorem}

\begin{proof}
The hardness follows immediately from Theorem~\ref{thr:npcomplete} as we can
represent sequence $s$ as a serial episode $H$. Then
$s$ covers $G$ if and only if $G \preceq H$.
\end{proof}

As mentioned in the introduction, pattern explosion is \emph{the} problem with discovering
general episodes. We tackle this by mining only closed episodes.

\begin{definition}
An episode $G$ is closed if there is no $H \succ G$ with $\freq{G} =
\freq{H}$.
\end{definition}

We should point out that, unlike with itemsets, an episode may have several
maximal superepisodes having the same frequency. Consider $G_1$, $G_3$, and
$G_4$ in Figure~\ref{fig:closure}, sequence $abcbdacbcd$ and window size
$\rho = 5$.  The support of episodes $G_1$, $G_3$ and $G_4$ is $2$. Moreover,
there is no superepisode of $G_3$ or $G_4$ that has the same support.
Hence, $G_3$ and $G_4$ are both maximal superepisodes having the same support
as $G_1$. This implies that we cannot define a closure operator based on
frequency. However, we will see in the next section that we can define a
closure based on instances. This closure, while not removing all redundant
episodes, will prune the search space dramatically. The final pruning will then be
done in a post-processing step.

\begin{figure}[htb!]
\centering

\begin{tikzpicture}
\node[node, anchor = base] (n1) {$a$};
\node[node, right = 0.7cm of n1.base, anchor = base] (n2) {$b$};
\node[node, below = 0.7cm of n1.base, anchor = base] (n3) {$c$};
\node[node, right = 0.7cm of n3.base, anchor = base] (n4) {$d$};
\node[labnode, left = 0cm of n3.base west, anchor = base east] (lab) {$G_1$};
\path [pedge] (n1) -- (n4);
\end{tikzpicture}\hfill
\begin{tikzpicture}
\node[node, anchor = base] (n1) {$a$};
\node[node, right = 0.7cm of n1.base, anchor = base] (n2) {$b$};
\node[node, below = 0.7cm of n1.base, anchor = base] (n3) {$c$};
\node[node, right = 0.7cm of n3.base, anchor = base] (n4) {$d$};
\node[labnode, left = 0cm of n3.base west, anchor = base east] (lab) {$G_2$};
\path [pedge] (n1) -- (n4);
\path [pedge] (n1) -- (n3);
\path [pedge] (n1.mid east) -- (n2.mid west);
\path [pedge] (n3.mid east) -- (n4.mid west);
\path [pedge] (n2) -- (n4);
\end{tikzpicture}\hfill
\begin{tikzpicture}
\node[node, anchor = base] (n1) {$a$};
\node[node, right = 0.7cm of n1.base, anchor = base] (n2) {$b$};
\node[node, below = 0.7cm of n1.base, anchor = base] (n3) {$c$};
\node[node, right = 0.7cm of n3.base, anchor = base] (n4) {$d$};
\node[labnode, left = 0cm of n3.base west, anchor = base east] (lab) {$G_3$};
\path [pedge] (n1) -- (n4);
\path [pedge] (n2) -- (n3);
\path [pedge] (n1) -- (n3);
\path [pedge] (n1.mid east) -- (n2.mid west);
\path [pedge] (n3.mid east) -- (n4.mid west);
\path [pedge] (n2) -- (n4);
\end{tikzpicture}\hfill
\begin{tikzpicture}
\node[node, anchor = base] (n1) {$a$};
\node[node, right = 0.7cm of n1.base, anchor = base] (n2) {$b$};
\node[node, below = 0.7cm of n1.base, anchor = base] (n3) {$c$};
\node[node, right = 0.7cm of n3.base, anchor = base] (n4) {$d$};
\node[labnode, left = 0cm of n3.base west, anchor = base east] (lab) {$G_4$};
\path [pedge] (n1) -- (n4);
\path [pedge] (n3) -- (n2);
\path [pedge] (n1) -- (n3);
\path [pedge] (n1.mid east) -- (n2.mid west);
\path [pedge] (n3.mid east) -- (n4.mid west);
\path [pedge] (n2) -- (n4);
\end{tikzpicture}\hfill

\caption{Toy episodes demonstrating closure.} 
\label{fig:closure}
\end{figure}
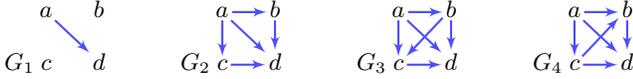

Our final step is to define transitively closed episodes that we will use along
with the instance-closure (defined in the next section) in order to reduce the
pattern space. 

\begin{definition}
Let $G$ be an episode. A \emph{transitive closure}, $\tcl{G}$, is obtained by
adding edges from a node to each of its descendants making the edge proper
if the descendant is proper, and weak otherwise. If $G = \tcl{G}$ we say that
$G$ is transitively closed.
\end{definition}

It is trivial to see that given an episode $G$, we have $G \sim \tcl{G}$.
Thus we can safely ignore all episodes that are not transitively closed.  From
now on, unless we state otherwise, all episodes are assumed to be transitively
closed.

\section{Handling Episode Instances}
\label{sec:instance}
The reason why depth-first search is efficient for itemsets is that at each
step we only need to handle the current projected dataset. In our setup we
have only one sequence so we need to transport the sequence into a more
efficient structure.

\begin{definition}
Given an input sequence $s$ and an episode $G$, an \emph{instance} $i$ is a
valid mapping from $G$ to $s$ such that for each $e \in \range{i}$ there is no
$f \in s - \range{i}$ such that $\lab{e} = \lab{f}$, $\ts{e} = \ts{f}$ and
$\id{f} < \id{e}$.
We define $\first{i} = \min \ts{i(n)}$ to
be the smallest time stamp and $\last{i} = \max \ts{i(n)}$ to be the largest
time stamp in $i$.  We require that $\last{i} - \first{i} \leq \rho - 1$, where $\rho$ is the size of the sliding window.
An \emph{instance set} of an episode $G$, defined as $\inst{G}$ is a set of
all instances ordered by $\first{i}$.  
\end{definition}

The condition in the definition allows us to ignore some redundant mappings
whenever we have two sequence events, say $e$ and $f$, with $\lab{e} = \lab{f}$
and $\ts{e} = \ts{f}$. If an instance $i$ uses only $e$, then we can obtain
$i'$ from $i$ by replacing $e$ with $f$. However, $i$ and $i'$ are essentially
the same for our purposes, so we can ignore either $i$ or $i'$.  We require
the instance set to be ordered so that we can compute the support efficiently.
This order is not necessarily unique.

Consider sequence $abcbdacbcd$ and $G_1$ in Figure~\ref{fig:closure}.
Then
$\inst{G_1} = ((1, 2, 3, 5), (1, 4, 3, 5), (6, 7, 8, 10), (6, 9, 8, 10))$\footnote{For simplicity, we write mappings as tuples}.

Using instances gives us several advantages. Adding new events and
edges to episodes becomes easy. For example, adding a proper edge $(n, m)$ is equivalent to
keeping instances with $\ts{i(n)} < \ts{i(m)}$. We will also compute support
and closure efficiently. We should point out that $\inst{G}$ may contain
an exponential number of instances, otherwise Theorem~\ref{thr:npcomplete}
would imply that $\p = \np$. However, this is not a problem in practice.

The depth-first search described in Section~\ref{sec:algorithm} adds events to
the episodes. Our next goal is to define algorithms for computing the resulting
instance set whenever we add an episode event, say $f$, to an episode $G$. Given
an instance $i$ of $G$ and a sequence event $e$ we will write $i + e$ to mean an
expanded instance by setting $i(f) = e$.

Let $G$ be an episode and let $I = \inst{G}$ be the instance set.
Assume a node $n \in V(G)$ and a label $l$. Let $H$ be the episode obtained
from $G$ by adding an episode event with  label $l$ to node $n$.
We can compute $\inst{H}$ from $\inst{G}$ using the \textsc{AugmentEqual}
algorithm given in Alg.~\ref{alg:augmentequal}.

\begin{algorithm}[htb!]
	\Input{$I = \inst{G}$, a node $n$, a label $l$}
	\Out{$\inst{\text{$G$ with $n$ augmented with $l$}}$}
	\Return $\set{i + e \mid \begin{array}{l} i \in I , e \in s, \lab{e} = l, \\ \ts{n} = \ts{e}, i + e \text{ is an instance} \end{array}}$\;
\caption{\textsc{AugmentEqual}$(I, n, l)$, augments $I$}
\label{alg:augmentequal}
\end{algorithm}

The second augmentation algorithm deals with the case where we are adding a new node with
an single event labelled as $l$ to a parallel episode. Algorithm \textsc{Augment}, given in Alg.~\ref{alg:augment}, computes the new
instance set. The algorithm can be further optimised by doing augmentation
with a type of merge sort so that post-sorting is not needed.

\begin{algorithm}[htb!]
\Input{$I = \inst{G}$, label $l$ of the new event}
\Out{$\inst{\text{$G$ with a new node labelled with $l$}}$}
	$E \define \set{e \in s \mid \lab{e} = l}$\; 
	$J \define \set{i + e \mid i \in I, e \in E, i + e \text{ is an instance}}$\;
	Sort $J$ by $\first{i}$\;
	\Return $J$\;
\caption{\textsc{Augment}$(I, l)$, augments instances}
\label{alg:augment}
\end{algorithm}

Our next step is to compute the support of an episode $G$ from $\inst{G}$.  We do
this with the \textsc{Support} algorithm, given in Alg.~\ref{alg:support}. The algorithm is based on the observation that there
are $\rho - (\last{i} - \first{i})$ windows that contain the instance $i$.
However, some windows may contain more than one instance and we need to 
compensate for this.

\begin{algorithm}[htb!]
	\Input{$I = \inst{G}$}
	\Out{$\freq{G}$}
	$J \define \emptyset$; $l \define \infty$\;
	\ForEach {$i \in I$ in reverse order} {
		\If {$\last{i} < l$} {
			Add $i$ to $J$\;
			$l \define \last{i}$\;
		}
	}
	$l \define -\infty$; $f \define 0$\;
	\ForEach {$i \in J$} {
		$d \define \rho - (\last{i} - \first{i})$\;
		$a \define 1 + \last{i} - \rho$\;
		$d \define d - \max(0, 1 + l - a)$\;\nllabel{alg:support:sum}
		$f \define f + d$\;
		$l \define \first{i}$\;
	}
	\Return $f$\;
\caption{\textsc{Support}$(I)$, computes support}
\label{alg:support}
\end{algorithm}

\begin{theorem}
$\textsc{Support}(\inst{G})$ computes $\freq{G}$.
\end{theorem}

\begin{proof}
Since $I$ is ordered by $\first{i}$, the first for loop of the algorithm
removes any instance for which there is an instance $j$ such that $\first{i}
\leq \first{j} \leq \last{j} \leq \last{i}$. In other words, any window
that contains $i$ will also contain $j$. We will show that the next for loop
counts the number of windows containing at least one instance from $W$, this will imply
the theorem.

To that end, let $i_n \in J$ be the $n$th instance in $J$ and define $S_n$ to
be the set of windows of size $\rho$ containing $i_n$. It follows that
$\abs{S_n} = \rho - (\last{i_n} - \first{i_n})$ and that the first window of
$S_n$ starts at $a_n = 1 + \last{i_n} - \rho$.
 Let $C_n = \bigcup_{m = 1}^n S_m$.
  Note that because of the pruning we have $a_{n - 1} < a_n$, this implies
that $C_{n - 1} \cap S_n = S_{n - 1} \cap S_n$.  We know that $\abs{S_{n - 1}
\cap S_n} = \max\fpr{0, 1 + \first{i_{n - 1}} - a_n}$.  This implies that on
Line~\ref{alg:support:sum} we have $d = \abs{S_n - S_{n - 1}}$ and since
$\abs{C_n} = \abs{C_{n - 1}} + d$, this proves the theorem.
\end{proof}

Finally, we define a closure episode of an instance set.

\begin{definition}
Let $I = \inst{G}$ be a set of instances. 
We define an \emph{instance-closure}, $H = \icl{I}$ to
be the episode having the same nodes and events as $G$. We define the edges
\[
\begin{split}
    \pe{H} & = \set{(a, b) \mid \ts{i(a)} < \ts{i(b)} \ \forall i \in I} \text{ and} \\
	\we{H} & = \set{(a, b) \mid \ts{i(a)} \leq \ts{i(b)}\  \forall i \in I} - \pe{H}.
\end{split}
\]
If $G = \icl{I}$ we say that $G$ is $i$-closed.
\end{definition}

Consider sequence $abcbdacbcd$ and episodes given in Figure~\ref{fig:closure}.
Since events $b$ and $c$ always occur between $a$ and $d$ in $\inst{G_1}$,
the instance closure is $\icl{\inst{G_1}} = G_2$. Note that $G_2$ is not closed
because $G_3$ and $G_4$ are both superepisodes of $G_2$ with the same support.
However, instance-closure reduces the search space dramatically because we do not have to
iterate the edges implied by the closure. Note that the closure may
produce cycles with weak edges. However, we will later show that we can ignore such
episodes.

\section{Discovering Episodes}
\label{sec:algorithm}
We are now ready to describe the mining algorithm. Our approach is a
straightforward depth-first search. The algorithm works on three different
levels.

The first level, consisting of \textsc{Mine} (Alg.~\ref{alg:mine}), and
\textsc{MineParallel} (Alg.~\ref{alg:mineparallel}), adds episode events.
\textsc{Mine} is only used for creating singleton episodes while
\textsc{MineParallel} provides the actual search. The algorithm adds
events so that the labels of the nodes are decreasing, 
$\lab{\node{G, i + 1}} \lexleq \lab{\node{G, i}}$. The search space is
traversed by either adding an event to the last node or by creating a
node with a single event. Episodes $G_1, \ldots, G_5$ in Figure~\ref{fig:state}
are created by the first level. The edge in $G_5$ is augmented by the closure.  

\tikzstyle{statestep} = [draw, -latex', thick, orange!70]
\tikzstyle{state} = [fill=orange!20, draw=orange!70,  inner sep = 1pt, rounded corners]

\begin{figure}[htb!]
\centering

\begin{tikzpicture}
\node[node, anchor = base] (g1n1) {$\emptyset$};

\node[node, below = 0.8cm of g1n1.base, anchor = base] (g2n1) {$G_1: a : 4$};

\node[node, right = 0.4cm of g2n1.base east, anchor = base west] (g3n1) {$G_2: aa : 2$};

\node[node, right = 0.9cm of g1n1.base east, anchor = base west] (g4n1) {$G_3: b : 2$};

\node[node, right = 0.7cm of g4n1.base east, anchor = base west] (g6n1) {$G_4: b$};
\node[node, right = 0.4cm of g6n1.base east, anchor = base west] (g6n2) {$a : 2$};

\node[node, below left = 0.8cm and 2pt of g6n1.base east, anchor = base east] (g7n1) {$G_5: b$};
\node[node, right = 0.4cm of g7n1.base east, anchor = base west] (g7n2) {$aa : 1$};
\path [pedge] (g7n2.mid west) -- (g7n1.mid east);

\node[node, right = 0.6cm of g6n2.base east, anchor = base west] (g8n1) {$G_6: b$};
\node[node, right = 0.4cm of g8n1.base east, anchor = base west] (g8n2) {$a : 1$};
\path [pedge] (g8n2.mid west) -- (g8n1.mid east);

\node[node, below = 0.8cm of g8n1.base, anchor = base] (g9n1) {$G_7: b$};
\node[node, right = 0.4cm of g9n1.base east, anchor = base west] (g9n2) {$a : 1$};
\path[pedge] (g9n1.mid east) -- (g9n2.mid west);

\begin{pgfonlayer}{background}

\node[state, fit=(g6n1) (g6n2)] (g6) {};
\node[state, fit=(g7n1) (g7n2)] (g7) {};
\node[state, fit=(g8n1) (g8n2)] (g8) {};
\node[state, fit=(g9n1) (g9n2)] (g9) {};
\node[state, fit=(g1n1)] (g1) {};
\node[state, fit=(g2n1)] (g2) {};
\node[state, fit=(g3n1)] (g3) {};
\node[state, fit=(g4n1)] (g4) {};

\path[statestep] (g1) -- (g2);
\path[statestep] (g2) -- (g3);
\path[statestep] (g1) -- (g4);
\path[statestep] (g4) -- (g6);
\path[statestep] (g6) -- (g8);
\path[statestep] (g6) -- (g9);
\path[statestep] (g6) -- (g7);

\end{pgfonlayer}
\end{tikzpicture}

\caption{Search space for a sequence $(aa)ba$, window size $\rho = 2$, and support threshold $\sigma = 2$. Each state shows the corresponding episode and its support.}
\label{fig:state}
\end{figure}
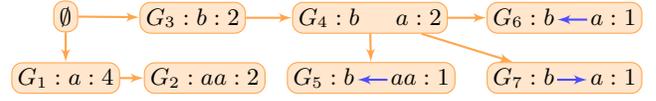

\begin{algorithm}
	\For {$x \in \Sigma$} {
		$I \define \inst{\text{singleton episode with the label }x}$\;
		$G \define \textsc{TestEpisode}(I, \emptyset, \emptyset)$\;
		\lIf {$G \neq \Null$} {
			$\textsc{MineParallel}(I, G)$;
		}
	}
\caption{$\textsc{Mine}$, discovers frequent closed episodes}
\label{alg:mine}
\end{algorithm}

\begin{algorithm}[htb!]
	\Input{episode $G$, $I = \inst{G}$}
	$\textsc{MineWeak}(I, G, \emptyset)$\;
	$M \define \abs{V(G)}$\;
	$n \define \node{G, M}$\;
	\For {$x \in \Sigma, x \geq \max \lab{n}$} {
		\If {$M = 1$ \OR $\lab{n}  \cup \set{x} \lexleq \lab{\node{G, M - 1}}$} {
			$J \define \textsc{AugmentEqual}(I, n, x)$\;
			$H \define \textsc{TestEpisode}(J, \emptyset, \emptyset)$\;
			\If {$H \neq \Null$} {
				$\textsc{MineParallel}(J, H)$\;
			}
		}
	}

	\For {$x \in \Sigma, x \leq \min \lab{n}$} {
		$J \define \textsc{Augment}(I, x)$\;
		$H \define \textsc{TestEpisode}(J, \emptyset, \emptyset)$\;
		\If {$H \neq \Null$} {
			$\textsc{MineParallel}(J, H)$\;
		}
	}
\caption{$\textsc{MineParallel}(I, G)$, recursive routine adding episode events}
\label{alg:mineparallel}
\end{algorithm}

The second level, \textsc{MineWeak}, given in Alg.~\ref{alg:mineweak}, adds weak edges
to the episode, while the third level, \textsc{MineProper}, given in
Alg.~\ref{alg:mineproper}, turns weak edges into proper edges. Both algorithms
add only those edges which keep the episode transitively closed. The algorithms keep a list
of forbidden weak edges $W$ and forbidden proper edges $P$. These lists guarantee that
each episode is visited only once.

$G_6$ and $G_7$ in Figure~\ref{fig:state}
are discovered by \textsc{MineWeak}, however, the weak edges are converted into proper edges
by the instance-closure. 

\begin{algorithm}[htb!]
	\Input{ep. $G$, $I = \inst{G}$, forbidden weak edges $W$}
	$A \define \set{(a, b) \mid a, b \in V(G), (a, b) \notin E(G)}$\;
	$\textsc{MineProper}(I, G, A, \emptyset)$\;
	\For {$(a, b) \notin E(G) \cup W$} {
		\If {$G + (a, b)$ is transitively closed } {
			$J \define \set{i \in I \mid \ts{i(a)} \leq \ts{i(b)}}$\;
			$H \define \textsc{TestEpisode}(J, W, \emptyset)$\;
			\If {$H \neq \Null$} {
				$\textsc{MineWeak}(J, H, W)$\;
			}
			Add $(a, b)$ to $W$\;
		}
	}
\caption{$\textsc{MineWeak}(I, G, W)$, recursive routine adding weak edges}
\label{alg:mineweak}
\end{algorithm}

\begin{algorithm}[htb!]
	\Input{episode $G$, $I = \inst{G}$, forbidden weak edges $W$, forbidden proper edges $P$}
	\For {$(a, b) \in \we{G} - P$} {
		\If {$G\,+$ proper edge $(a, b)$ is transitively closed } {
			$J \define \set{i \in I \mid \ts{i(a)} < \ts{i(b)}}$\;
			$H \define \textsc{TestEpisode}(J, W, P)$\;
			\If {$H \neq \Null$} {
				$\textsc{MineProper}(J, H, W, P)$\;
			}
			Add $(a, b)$ to $P$\;
		}
	}
\caption{$\textsc{MineProper}(I, G, W, P)$, recursive routine adding proper edges}
\label{alg:mineproper}
\end{algorithm}

At each step, we call \textsc{TestEpisode}. This routine, given a set of
instances $I$, will compute the instance-closure $\icl{I}$ and test it. If the
episode passes all the tests, the algorithm will return the episode and the
search is continued, otherwise the branch is terminated.  There are four
different tests. The first test checks whether the episode is frequent.  If we
pass this test, we compute the instance-closure $H = \icl{I}$. The second test checks whether $H$
contains cycles.  Let $G$ be an episode such that $I = \inst{G}$. In order for $H$
to have cycles we must have two nodes, say $n$ and $m$, such that $\ts{i(n)} =
\ts{i(m)}$ for all $i \in I$. Let $G'$ be an episode obtained from $G$ by
merging nodes in $n$ and $m$ together. We have $G \preceq G'$ and
$\freq{G} = \freq{G'}$. This holds for any subsequent episode
discovered in the branch allowing us to ignore the whole branch.

If the closure introduces into $H$ any edge that has been explored in the
previous branches, then that implies that $H$ has already been discovered. Hence, we
can reject $H$ if any such edge is introduced. The final condition is that during
\textsc{MineProper} no weak edges should be added into $H$ by the closure.
If a weak edge is added, we can reject $H$ because it can be reached via an alternative route, by letting
\textsc{MineWeak} add the additional edges and then calling \textsc{MineProper}.

The algorithm keeps a list $\efam{C}$ of all discovered episodes that are
closed.  If all four tests are passed, the algorithm tests whether there are
subepisodes of $G$ in $\efam{C}$ having the same frequency, and deletes them.
On the other hand, if there is a superepisode of $G$ in $\efam{C}$, then
$G$ is not added into $\efam{C}$.

\begin{algorithm}[htb!]
	\Input{$I = \inst{G}$, forbidden weak edges $W$, forbidden proper edges $P$}
	\Out{$\icl{I}$, if $\icl{I}$ passes the tests, \Null otherwise}
	$f \define \textsc{Support}(I)$\;
	\lIf {$f < \sigma$} {
		\Return \Null;
	}
	$G \define \icl{I}$\;
	\If {there are cycles in $G$} {
		\Return \Null\;
	}
	\If {$\we{G} \cap W \neq \emptyset$ \OR $\pe{G} \cap P \neq \emptyset$} {
		\Return \Null\;
	}
	$\freq{G} \define f$\;
	\ForEach {$H \in \efam{C}, \lab{H} \cap \lab{G} \neq \emptyset$} {
		\If {$\freq{G} = \freq{H}$ \AND $G \preceq H$} {
			\Return $G$\;
		}
		\If {$\freq{G} = \freq{H}$ \AND $H \preceq G$} {
			Delete $H$ from $\efam{C}$\;
		}
	}
	Add $G$ to $\efam{C}$\;
	\Return $G$\;
\caption{$\textsc{TestEpisode}(I, W, P)$, tests the episode $\icl{I}$ and updates $\efam{C}$, the list of discovered episodes}
\label{alg:testepisode}
\end{algorithm}

\section{Testing Subepisodes}
\label{sec:compute}
In this section we will describe the technique for computing $G \preceq H$.  In
general, this problem is difficult to solve as pointed out by
Theorem~\ref{thr:hardness}.  Fortunately, in practice, a major part of the
comparisons are easy to do.  The following theorem says that if the labels of
$G$ are unique in $H$, then we can easily compare $G$ and $H$.

\begin{theorem}
\label{thr:easy}
Assume two episodes $G$ and $H$. Assume that $\lab{G}
\subset \lab{H}$ and for each event $e$ in $G$ only one event occurs in $H$
with the same label. Let $\pi$ be the unique mapping from episode events in $G$
to episode events in $H$ honouring the labels.  Then $G \preceq H$ if and only if

\begin{enumerate*}
\item $\node{e} = \node{f}$ implies that $\node{\pi(e)} = \node{\pi(f)}$,
\item $\node{e}$ is a proper child of $\node{f}$ implies that
$\node{\pi(e)}$ is a proper child of $\node{\pi(f)}$,
\item $\node{e}$ is a child of $\node{f}$ implies that $\node{\pi(e)}$ is a
child of $\node{\pi(f)}$ or $\node{\pi(e)} = \node{\pi(f)}$,
\end{enumerate*}
for any two events $e$ and $f$ in $G$.
\end{theorem}

If the condition in Theorem~\ref{thr:easy} does not hold we will have to resort
to enumerating the sequences covering $H$. In order to do that, we need to
extend the definition of coverage and subset relationship to the set of
episodes.

\begin{definition}
A sequence $s$ covers an episode set $\efam{G}$ if there is an episode $G \in
\efam{G}$ such that $s$ covers $G$. Given two episode sets $\efam{G}$ and $\efam{H}$
we define $\efam{G} \preceq \efam{H}$ if every sequence that covers $\efam{H}$ also
covers $\efam{G}$.
\end{definition}

We also need a definition of a prefix subgraph.
\begin{definition}
Given a graph $G$, a \emph{prefix subgraph} is a non-empty induced subgraph of $G$ with \emph{no} proper edges such that if
a node $n$ is included then the parents of $n$ are also included. Given a
multiset of labels $L$ and an episode $G$ we define $\pre{G, L}$ to be the set
of all maximal prefix subgraphs such that $\lab{V} \subseteq L$ for each $V \in
\pre{G, L}$. We define $\tail{G, L} = \set{G - V \mid V \in \pre{G, L}}$ to
be the episodes with the remaining nodes. Finally, given an episode set $\efam{G}$ we define
$\tail{\efam{G}, L} = \bigcup_{G \in \efam{G}} \tail{G, L}$.
\end{definition}

\begin{example}
Consider episodes given in Figure~\ref{fig:tail}. We have $\tail{G, ab} =
\set{H_1, H_2, H_3}$ and $\tail{G, a} = \set{H_1}$.
\end{example}

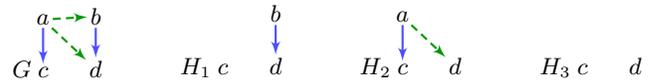
\begin{figure}[h!]
\begin{tikzpicture}
\node[node, anchor = base] (n1) {$a$};
\node[node, right = 0.7cm of n1.base, anchor = base] (n2) {$b$};
\node[node, below = 0.7cm of n1.base, anchor = base] (n3) {$c$};
\node[node, right = 0.7cm of n3.base, anchor = base] (n4) {$d$};
\node[labnode, left = 0cm of n3.base west, anchor = base east] (lab) {$G$};
\path [pedge] (n1) -- (n3);
\path [pedge] (n2) -- (n4);
\path [wedge] (n1) -- (n2);
\path [wedge] (n1) -- (n4);
\end{tikzpicture}\hfill
\begin{tikzpicture}
\node[node, anchor = base] (n1) {};
\node[node, right = 0.7cm of n1.base, anchor = base] (n2) {$b$};
\node[node, below = 0.7cm of n1.base, anchor = base] (n3) {$c$};
\node[node, right = 0.7cm of n3.base, anchor = base] (n4) {$d$};
\node[labnode, left = 0cm of n3.base west, anchor = base east] (lab) {$H_1$};
\path [pedge] (n2) -- (n4);
\end{tikzpicture}\hfill
\begin{tikzpicture}
\node[node, anchor = base] (n1) {$a$};
\node[node, below = 0.7cm of n1.base, anchor = base] (n3) {$c$};
\node[node, right = 0.7cm of n3.base, anchor = base] (n4) {$d$};
\node[labnode, left = 0cm of n3.base west, anchor = base east] (lab) {$H_2$};
\path [pedge] (n1) -- (n3);
\path [wedge] (n1) -- (n4);
\end{tikzpicture}\hfill
\begin{tikzpicture}
\node[node, anchor = base] (n1) {};
\node[node, below = 0.7cm of n1.base, anchor = base] (n3) {$c$};
\node[node, right = 0.7cm of n3.base, anchor = base] (n4) {$d$};
\node[labnode, left = 0cm of n3.base west, anchor = base east] (lab) {$H_3$};
\end{tikzpicture}\hfill

\caption{Toy episodes demonstrating $\tail{G, ab}$.}
\label{fig:tail}
\end{figure}

The main motivation for our recursion is given in the following theorem.

\begin{theorem}
\label{thr:recurse}
Given an episode set $\efam{G}$ and an episode $H$, $\efam{G} \preceq H$ if and only if for each prefix
subgraph $V$ of $H$, we have $\tail{\efam{G}, \lab{V}} \preceq H - V$. 
\end{theorem}

We focus for the rest of this section on implementing the recursion in
Theorem~\ref{thr:recurse}.  We begin by an algorithm, \textsc{Generate}, given
in Alg.~\ref{alg:generate}, that, given a graph without proper edges, discovers
all prefix subgraphs.

\begin{algorithm}[htb!]
	\Input{graph $G$, nodes $V$ discovered so far}
	\Out{list of nodes of all prefix subgraphs}
	$\efam{O} \define \emptyset$\;
	\ForEach {$n \in \sources{G}$} {
		$\efam{O} \define \efam{O} \cup \set{V + n} \cup \textsc{Generate}\fpr{G - n, V + n}$\;
		Remove $n$ and its descendants from $G$\;
	}
	\Return $\efam{O}$\;
\caption{$\textsc{Generate}\fpr{G, V}$, recursive routine for iterating the nodes of all prefix subgraphs of $G$}
\label{alg:generate}
\end{algorithm}

Given a prefix subgraph $V$ of $H$, our next step is to discover all maximal prefix subgraphs of $G$
whose label sets are subsets of $L = \lab{V}$.
The algorithm,
\textsc{Consume}, creating this list, is given in Alg.~\ref{alg:consume}.
\textsc{Consume} enumerates over all sources. For each source $n$ such that
$\lab{n} \subseteq L$, the algorithm tests if there is another node sharing a
label with $n$. If so, the algorithm creates an episode without $n$ and
its descendants and calls itself. This call produces a list of prefix graphs
$\efam{W}$ not containing node $n$. The algorithm removes all graphs from
$\efam{W}$ that can be augmented with $n$ (since in that case they are not maximal).
Finally, \textsc{Consume} adds $n$ to the current prefix subgraph,
removes $n$ from $G$, and removes $\lab{n}$ from $L$.

\begin{algorithm}[htb!]
	\Input{graph $G$, nodes $V$ discovered so far, $L$ label set}
	\Out{list of nodes of all maximal prefix subgraphs}
	$\efam{O} \define \emptyset$\;
	\While {$\sources{G} \neq \emptyset$} {
		$n \define$ a node from $\sources{G}$\;
		\If {$\lab{n} \nsubseteq L$} {
			Remove $n$ and its descendants from $G$\;
			\Continue\;
		}
		\If {there is $m \in V(G)$ s.t. $\lab{n} \cap \lab{m} \neq \emptyset$} {
			$H \define G$ with $n$ and its descendants removed\;
			$\efam{W} \define \textsc{Consume}\fpr{H, L, V}$\;
			$\efam{O} \define \efam{O} \cup \set{W \in \efam{W} \mid \lab{W} \cup \lab{n} \nsubseteq L}$\;
		}

		$V \define V + n$; $L \define L - \lab{n}$\;
		Remove $n$ from $G$\;
	}
	\Return $\efam{O} \cup \set{V}$\;
\caption{$\textsc{Consume}\fpr{G, L, V}$, recursive routine for iterating the nodes of all prefix subgraphs of $G$ whose labels are contained in $L$}
\label{alg:consume}
\end{algorithm}

We are now ready to describe the recursion step of Theorem~\ref{thr:recurse}
for testing the subset relationship. The algorithm, \textsc{Step}, is given in
Alg.~\ref{alg:step}. Given an episode set $\efam{G}$ and a graph $H$, the
algorithm computes $\efam{G} \preceq H$. First, it tests whether we can apply
Theorem~\ref{thr:easy}.  If this is not possible, then the algorithm first
removes all nodes from $H$ not carrying a label from $\lab{G}$ for $G \in
\efam{G}$. This is allowed because of the following lemma.

\begin{lemma}
\label{lem:reduce}
Let $G$ and $H$ be two episodes. Let $n \in V(H)$ be a node
such that $\lab{n} \cap \lab{G} = \emptyset$.  Let $H'$ be the episode obtained
from $H$ by removing $n$.  Then $G \preceq H$ if and only if $G \preceq H'$.
\end{lemma}

\textsc{Step} continues by creating a subgraph $Y$ of $H$ containing only
proper sources.  The algorithm generates all prefix subgraphs $\efam{V}$ of $Y$
and tests each one. For each subgraph $V \in \efam{V}$, \textsc{Step} calls
\textsc{Consume} and builds an episode set $\efam{T} = \tail{\efam{G},
\lab{V}}$. The algorithm then calls itself recursively with
$\textsc{Step}(\efam{T}, H - V)$.  If at least one of these calls fails, then we know
that $\efam{G} \npreceq H$, otherwise $\efam{G} \preceq H$.

\begin{algorithm}[tb!]
	\Input{episode set $\efam{G}$ and an episode $H$}
	\Out{$\efam{G} \preceq H$}

	\ForEach {$G \in \efam{G}$} {
		\If {Theorem~\ref{thr:easy} guarantees that $G \preceq H$} {
			\Return \True\;
		}
		\If {Theorem~\ref{thr:easy} states that $G \npreceq H$ \AND $\abs{\efam{G}} = 1$} {
			\Return \False\;
		}
	}

	$V(H) \define \set{n \in V(H) \mid \lab{n} \cap \lab{G} \neq \emptyset, G \in \efam{G}}$\;
	$Y \define $ subgraph of $H$ with only proper sources\; 
	$\efam{V} \define \textsc{Generate}(Y, \emptyset)$\;
	\ForEach {$V \in \efam{V}$} {
		$F \define H - V$; $\efam{T} \define \emptyset$\;

		\ForEach {$G \in \efam{G}$} {
			$X \define $ subgraph of $G$ with only proper sources\; 
			$\efam{W} \define \textsc{Consume}(X, \lab{V}, \emptyset)$\;
			$\efam{T} \define \efam{T} \cup \set{G - W \mid W \in \efam{W}}$\;
		}

		\If {$\textsc{Step}(\efam{T}, F) = \False$} {
			\Return \False\;
		}
	}
	\Return \True\;
\caption{$\textsc{Step}\fpr{\efam{G}, H}$, recursion solving $\efam{G} \preceq H$}
\label{alg:step}
\end{algorithm}

Finally, to compute $G \preceq H$, we simply call $\textsc{Step}(\set{G}, H)$.

\section{Experiments}
\label{sec:exp}

We begin our experiments with a synthetic dataset. Our goal is to demonstrate the
need for using the closure. In order to do that we created sequences with a planted
pattern $(s_1s_2)(s_3s_4)\cdots(s_{2N - 1}s_{2N})$. We added this pattern $100$
times 50 time units apart from each other. We added $500$ noise events uniformly
spreading over the whole sequence. We sampled the labels for the noise events
uniformly from $900$ different labels. The labels for the noise and the labels
of the pattern were mutually exclusive. We varied $N$ from $1$ to $7$.  We ran
our miner using a window size of $\rho = 10$ and varied the support threshold
$\sigma$.  The results are given in Table~\ref{tab:synthetic}.

\begin{table}[htb!]
\caption{Results from synthetic sequences with a planted episode $P$.
The columns show the number of nodes in $P$, the support
threshold, the size of the final output, the
number of instance-closed episodes, the number of
subepisodes of $P$, and the number of sequence scans, respectively.}
\vspace{0.1cm}
\centering
\begin{tabular}{lrrrrr}
\toprule
$\abs{V(P)}$ & $\sigma$ & $\abs{\efam{C}}$ & $i$-closed & frequent & scans\\
\midrule
1 & 100 & 1 & 3 & 3 & 46\\
2 & 100 & 3 & 15 & 27 & 200 \\
3 & 100 & 6 & 63 & 729 & 744 \\
4 & 100 & 10 & 255 & $59\,049$ & $3\,964$ \\
5 & 100 & 15 & $1\,023$ & $14\,348\,907$ & $11\,123$\\
6 & 100 & 21 & $4\,095$ & $\approx 10^{10}$ & $48\,237$ \\
7 & 100 & 28 & $16\,383$ & $\approx 2 \times 10^{13}$ & $191\,277$  \\
7 & 50  & 29 & $16\,384$ & $>2 \times 10^{13}$ & $191\,243$  \\
7 & 40  & 32 & $16\,387$ & $>2 \times 10^{13}$ & $191\,298$  \\
7 & 30  & 39 & $16\,394$ & $>2 \times 10^{13}$ & $191\,463$  \\
7 & 20  & 127 & $16\,488$ & $>2 \times 10^{13}$ & $197\,773$  \\
7 & 10  & 684 & $52\,297$ & $>2 \times 10^{13}$ & $480\,517$  \\
\bottomrule
\end{tabular}
\label{tab:synthetic}
\end{table}

When we are using a support threshold of $100$, the only closed frequent patterns
are the planted pattern and its subpatterns of form $(s_is_{i +
1})\cdots(s_js_{j + 1})$, since the frequency of these subpatterns is slightly
higher than the frequency of the whole pattern. The number of instance-closed episodes (given in
the $4^{\text{th}}$ column of Table~\ref{tab:synthetic}) grows more rapidly. The reason for
this is that the instance-closure focuses on the edges and does not add any new
nodes. However, this ratio becomes a bottleneck only when we are dealing with
extremely large serial episodes, and for our real-world datasets this ratio is
relatively small. 

The need for instance-closure becomes apparent when the number of instance-closed episodes is
compared to the number of all possible general subepisodes (including those that are not tranistively closed) of a planted
pattern, given in the $5^{\text{th}}$ column of Table~\ref{tab:synthetic}. We see that had
we not used instance-closure, a single pattern having 6 nodes and 12 events
renders pattern discovery infeasible.

As we lower the threshold, the number of instance-closed episodes and closed
episodes increases, however the ratio between instance-closed and closed
episodes improves. The reason for this is that the output contains episodes other than our planted episode, and those mostly contain
a small number of nodes. 

We also measured the number of sequence scans, namely the number of calls made to
\textsc{Support}, to demonstrate how fast the negative border is growing. Since
our miner is a depth-first search and the computation of frequency is based on
instances, a single scan is fast, since we do not have to scan the whole
sequence.

Our second set of experiments was conducted on real-world data. The dataset
consists of alarms generated in a factory, and contains $514\,502$ events of $9\,595$
different types, stretching over 18 months. An entry in the dataset consists of
a time stamp and an event type.  Once again, we tested our algorithm at various
thresholds, varying the window size from 3 to 15 minutes. Here, too, as shown
in Table~\ref{tab:alarms}, we can see that the number of $i$-closed episodes
does not explode to the level of all frequent episodes, demonstrating the need
for using the $i$-closure as an intermediate step. Furthermore, we see that in
a realistic setting, the number of $i$-closed episodes stays closer to the
number of closed episodes than in the above-mentioned synthetic dataset. This
is no surprise, since real datasets tend to have a lot of patterns containing a small
number of events.

As the discovery of all frequent episodes is infeasible, we estimated their number as follows. An episode $G$ has
$a(G) = 3^{\abs{\pe{G}}}2^\abs{\we{G}}$ subepisodes (including those that are not transitively closed) with the same
events and nodes. From the discovered $i$-closed episodes we selected a subset
$\efam{G}$ such that each $G \in \efam{G}$ has a unique set of events and a
maximal $a(G)$. Then the lower bound for the total number of frequent episodes is $\sum_{G \in
\efam{G}} a(G)$. Using such a lower bound is more than enough to confirm that the number of frequent episodes explodes much faster than the number of closed and $i$-closed episodes, as can be seen in Table~\ref{tab:alarms}.

Furthermore, our output contained a considerable number of episodes with
simultaneous events --- patterns that no existing method would have
discovered. The runtimes ranged from just under 2 seconds to 90 seconds for the
lowest reported thresholds.

\begin{table}[htb!]
\caption{Results from the alarms dataset.}
\vspace{0.1cm}
\centering
\begin{tabular}{lrrrrr}
\toprule
win (s) & $\sigma / 10^3$ & $\abs{\efam{C}}$ & $i$-closed & freq.(est) & scans\\
\midrule
180 & $500$ & 6 & 6 & 6 & 194\\
180 & $400$ & 8 & 8 & 8 & 220\\
180 & $300$ & 12 & 12 & 12 & 282\\
180 & $240$ & 23 & 26 & 26 & 792\\
\midrule
600 & $2\,000$ & 4 & 4 & 4 & 128\\
600 & $1\,000$ & 24 & 27 & $39$ & 374\\
600 & $500$ & 90 & 137 &  493 & $1\,196$\\
600 & $280$ & 422 & 698 & $2\,321$ & $8\,758$\\
\midrule
900 & $2\,000$ & 24 & 26 & $40$ & 350\\
900 & $1\,000$ & 52 & 58 & $94$ & 745\\
900 & $500$ & 280 & 426 & $1\,997$ & $4\,604$\\
900 & $350$ & $1\,845$ & $9\,484$ & $190\,990$ & $63\,735$\\
\bottomrule
\end{tabular}
\label{tab:alarms}
\end{table}

Our third dataset consisted of trains delayed at a single railway station in Belgium. The dataset consists of actual departure times of delayed trains, coupled with train numbers, and contains $10\,115$ events involving $1\,280$ different train IDs, stretching over a period of one month. A window of 30 minutes was chosen by a domain expert. The time stamps were expressed in seconds, so a single train being delayed on a particular day would be found in 1800 windows. Therefore, the frequency threshold for interesting patterns had to be set relatively high. The results are shown in Table~\ref{tab:trains}, and were similar to those of the alarm dataset. The runtimes ranged from a few milliseconds to 55 minutes for the lowest reported threshold. The largest discovered pattern was of size 10, and the total number of frequent episodes at the lowest threshold was at least $33$ million, once again demonstrating the need for both outputting only closed episodes and using instance-closure.

\begin{table}[htb!]
\caption{Results from the trains dataset.}
\vspace{0.1cm}
\centering
\begin{tabular}{lrrrr}
\toprule
$\sigma$ & $\abs{\efam{C}}$ & $i$-closed & freq.(est) & scans\\
\midrule
$30\,000$ & 141 & 141 & 141 & $9\,575$\\
$20\,000$ & $1\,994$ & $1\,995$ & $2\,593$ & $219\,931$\\
$17\,000$ & $8\,352$ & $8\,416$ & $22\,542$ & $812\,363$\\
$15\,000$ & $26\,170$ & $26\,838$ & $172\,067$ & $2\,231\,360$\\
$13\,000$ & $94\,789$ & $101\,882$ & $3\,552\,104$ & $6\,865\,877$\\
$12\,000$ & $189\,280$ & $211\,636$ & $33\,660\,094$ & $12\,966\,895$\\
\bottomrule
\end{tabular}
\label{tab:trains}
\end{table}

\section{Conclusions}
\label{sec:conclusions}
In this paper we introduce a new type of sequential pattern, a general episode
that takes into account simultaneous events, and provide an efficient
depth-first search algorithm for mining such patterns.

This problem setup has two major challenges. The first challenge is the pattern explosion
which we tackle by discovering only closed episodes. Interestingly enough, we cannot
define closure based on frequency, hence we define closure based on instances. While it holds that frequency-closed episodes are instance-closed, the
opposite is not true. However, in practice, instance-closure reduces search
space dramatically so we can mine all instance-closed episodes and discover 
frequency-closed episodes as a post-processing step.

The second challenge is to correctly compute the subset relationship between
two episodes. We argue that using a subset relationship based on graphs is not
optimal and will lead to redundant output. We define a subset relationship
based on coverage and argue that this is the correct definition. This
definition turns out to be \np-hard. However, this is not a problem since in
practice most of the comparisons can be done efficiently.

\section{Acknowledgments}
Nikolaj Tatti is supported by a Post-doctoral Fellowship of the Research Foundation -- Flanders (\textsc{fwo}).

The authors wish to thank Toon Calders for providing the proof that checking whether a sequence covers an episode is \textbf{NP}-hard on a small piece of
paper.

\bibliography{bibliography}

\pagebreak
\appendix
\section{Proofs of Theorems}
\begin{proof}[of Theorem~\ref{thr:npcomplete}]
If $s$ covers $G$, then the map $f$ mapping the nodes of $G$ to indices of $s$
provides the certificate needed for the verification. Hence, testing the
coverage is in \np.

In order to prove the completeness we reduce \textsc{3SAT} to the coverage. In
order to do so, given the formula $F$, we will build an episode $G$ and a
sequence such that sequence $s$ covers $G$ if and only if $F$ is
satisfiable.

Assume that we are given a formula $F$ with $M$ variables and $L$ clauses.  We
define the alphabet for the labels to be $\Sigma = \enset{\alpha_1}{\alpha_M}
\cup \enset{\beta_1}{\beta_L}$, where $\alpha_i$ is identified with the $i$th
variable and $\beta_j$ is identified with $j$th clause.

We will now construct $G$. The nodes in the episode consist of
three groups. The first group, $P = \enset{p_1}{p_M}$, contains $M$ nodes. A
node $p_i$ is labelled as $\lab{p_i} = \alpha_i$. These nodes represent the
positive instantiation of the variables. The second group of nodes, $N =
\enset{n_1}{n_M}$, also contains $M$ nodes, and the labels are again $\lab{n_i}
= \alpha_i$. These nodes represent the negative instantiation of the variables.
Our final group is $C = \enset{c_1}{c_{3L}}$, contains $3L$ nodes, $3$ nodes
for each clause. The labels for these nodes are $\lab{c_{3j - 2}} =
\lab{c_{3j -1}} = \lab{c_{3j}} = \beta_j$. The edges of the episode $G$ are as
follows: Let $\beta_j$ be the $j$th clause in $F$ and let $\alpha_i$ be the
$k$th variable occurring in that clause. Note that $k = 1, 2, 3$. We connect
$p_i$ to $c_{3j + k - 3}$ if the variable is positive in the clause. Otherwise,
we connect $n_i$ to $c_{3j + k - 3}$.

The sequence $s$ consists of $5$ consecutive subsequences $s = s_1s_2s_3s_4s_5$. We define $s_1 = s_3 = \alpha_1\cdots \alpha_M$
and $s_2 = s_4 = s_5 = \beta_1 \cdots \beta_L$, that is, $s$ is equal to
\[
\alpha_1 \cdots \alpha_M \beta_1 \cdots \beta_L \alpha_1 \cdots \alpha_M \beta_1 \cdots \beta_L \beta_1 \cdots \beta_L.
\]

Our final step is to prove that $s$ covers $G$ whenever $F$ is satisfiable.
First assume that $F$ is satisfiable and let $t_i$ be the truth assignment for
the variable $\alpha_i$. We need to define a mapping $f$. If $t_i$ is true, then we
map $p_i$ into the first group $s_1$ and $n_i$ into the third group $s_3$. If
$t_i$ is false, then we map $n_i$ into $s_1$ and $p_i$ into $s_3$. Since each
clause is now satisfied, among the nodes $c_{3j - 2}$, $c_{3j - 1}$, and
$c_{3j}$ there is at least one node, say $c_k$, such that the parent of that
node ($p_i$ or $n_i$) is mapped to the first group $s_1$. We can map $c_k$ into
the second group $s_2$. The remaining two nodes are mapped into the fourth and the
fifth group, $s_4$ and $s_5$. Clearly this mapping is valid since all the nodes
are mapped and the edges are honoured.

To prove the other direction, let $f$ be a valid mapping of $G$ into $s$.
Since the sequence has the same amount of symbols as there are nodes in the
graph the mapping $f$ is surjective.  Define a truth mapping by setting the
$i$th variable to true if $p_i$ occurs in $s_1$, and false otherwise. Each
symbol in the second group $s_2$ is covered. Select one symbol, say $\beta_j$
from that group. The corresponding node, say $c_k$, has a parent (either $p_i$
or $n_i$) that must be mapped into the first group. This implies that the
truth value for the $i$th variable satisfies the $j$th clause. Since all clauses
are satisfied, $F$ is satisfied. This completes the proof.
\end{proof}

\begin{proof}[of Theorem~\ref{thr:easy}]
To prove the theorem we will use the following straightforward lemma.

\begin{lemma}
\label{lem:tclinstance}
Let $G$ be a transitively closed episode. Let $e$ and $f$ be two nodes. Then
there exists a sequence $s$ covering $G$ and a mapping $m$ such that
\setlength{\leftmargin}{0pt}
\begin{enumerate*}
\item if $\node{e} \neq \node{f}$, then $\ts{m(e)} \neq \ts{m(f)}$,
\item if $\node{e}$ is not a child of $\node{f}$, then\\$\ts{m(e)} < \ts{m(f)}$,
\item if $\node{e}$ is not a proper child of $\node{f}$, then\\$\ts{m(e)} \leq \ts{m(f)}$,
\end{enumerate*}
\end{lemma}

The 'if' part is straightforward so we only prove the 'only if' part.  Assume
that $G \preceq H$ and let $e$ and $f$ be two events in $G$ violating one of
the conditions. Apply Lemma~\ref{lem:tclinstance} with $H$, $\pi(e)$, and
$\pi(f)$ to obtain a sequence $s$ and a mapping $m$. We can safely assume that
$m$ is surjective. Define $m'(x) = m(\pi(x))$ for an event $x$ in $G$.  Since
every label is unique in $s$, $m'$ is the only mapping that will honor the
labels in $G$. Now Lemma~\ref{lem:tclinstance} implies that $m'$ is not a
valid mapping for $G$, hence $G \npreceq H$, which proves the theorem.
\end{proof}

\begin{proof}[of Theorem~\ref{thr:recurse}]

Assume that the condition in the theorem holds and let $s$ be a sequence
covering $H$ and let $m$ be the corresponding mapping. We can safely assume
that $m$ is surjective.  Let $t$ be the smallest time stamp in $s$ and $E$
the events in $s$ having the time stamp $t$. Let $V$ be the corresponding
prefix subgraph, 
\[
	V = \set{n \in V(G) \mid \ts{m(n)} = t}.
\]
By assumption, we have $\tail{\efam{G}, \lab{V}} \preceq H - V$, hence there is
$G \in \efam{G}$ and $W \in \pre{G, L}$ such that $G - W \preceq H - V$ and
$\lab{W} \subseteq \lab{V}$. Let $s'$ be the sequence obtained from $s$ by
removing $E$. Since $s'$ covers $H - V$, there is a map $m'$ mapping $G - W$ to
$s'$. We can extend this map to $G$ by mapping $W$ to $E$. Thus $s$ covers $G$,
hence $s$ covers $\efam{G}$, and by definition $\efam{G} \preceq H$.

To prove the other direction, assume that $\efam{G} \preceq H$. Let $V$ be a
prefix subgraph of $H$. Let $s'$ be a sequence covering $H - V$. We can extend
this sequence to $s$ by adding the events with labels $\lab{V}$ in front of $s'$. Let us
denote these events by $E$.  Sequence $s$ covers $H$, hence it covers $\efam{G}$.

By assumption, there is a mapping $m'$ from $G$ to $s$ for some $G \in
\efam{G}$. Let $W$ be the (possibly empty) prefix graph of $G$ mapping to $E$.
We can assume that $W \in \pre{G, \lab{V}}$, that is, $W$ is maximal,
otherwise let $W' \in \pre{G, \lab{V}}$ such that $W \subset W'$. We can
modify $m'$ by remapping the nodes in $W' - W$ to the events in $E$. This will make $W$ equal to $W'$. Since
$s'$ covers $G - W$, it follows that $s'$ covers $\tail{\efam{G}, \lab{V}}$
which proves the theorem.
\end{proof}

\begin{proof}[of Lemma~\ref{lem:reduce}]
The 'if' part is trivial.  To prove the 'only if' part, assume that $G \preceq
H$. Let $s'$ be a sequence covering $H'$ and let $m'$ be the mapping. If we can
show that $s'$ can be extended to $s$ by adding events with the labels
$\lab{n}$ such that $s$ covers $G$, then the corresponding mapping $m$ will 
also be a valid mapping from $G$ to $s'$, which will prove the lemma.

If $n$ is a source or a sink (a node without children), then we can extend $s'$
by adding the event with the label $\lab{n}$ either at the beginning or at the
end of $s'$.

Assume that $n$ has parents and children.  Let $t_1$ be the largest time stamp
of the parents of $n$ and let $t_2$ be the smallest time stamp of the children
of $n$. Extend $s'$ to $s$ by adding an event $f$ with the label $\lab{n}$ and the time stamp $t
= 1/2(t_1 + t_2)$.  Let $x$ be a child of $n$ and $y$ a parent of $n$. Since
$H$ is transitively closed, $x$ is a child of $y$ in $H'$. Since this holds for
any $x$ and $y$, we have $t_1 \leq t_2$. Hence $\ts{y} \leq t \leq \ts{x}$.  If
$x$ is proper descendant of $n$, then it is also a proper descendant of any parent of
$n$, hence $t_1 < \ts{x}$ and consequently $t < \ts{x}$.  The same holds for $y$.
This implies that $s$ covers $G$.
\end{proof}

\end{document}